\documentclass[a4paper,11pt]{article}
\usepackage{amsmath,epsfig}
\usepackage{mathtools,amsmath,amsfonts,amssymb,amsthm,mathtools}
\usepackage{fullpage}
\usepackage{xspace}

\usepackage[footnotesize]{caption}
\usepackage[caption=false,font=footnotesize]{subfig}
\usepackage{graphicx}
\graphicspath{{./}{./figs/}}

\usepackage{pgfplots}

\newtheorem{Thm}{Theorem}[section]

\newtheorem{Propn}[Thm]{Proposition}

\renewcommand{\th}{%
    \ifmmode
        ^\mathrm{th}%
    \else%
        \textsuperscript{th}\xspace%
    \fi%
}
\newcommand{\ost}{%
    \ifmmode
        ^\mathrm{st}%
    \else%
        \textsuperscript{st}\xspace%
    \fi%
}
\newcommand{\rd}{%
    \ifmmode
        ^\mathrm{rd}%
    \else%
        \textsuperscript{rd}\xspace%
    \fi%
}
\newcommand{\nd}{%
    \ifmmode
        ^\mathrm{nd}%
    \else%
        \textsuperscript{nd}\xspace%
    \fi%
}

\newcommand{\setcst}[1]{{\cl W}({#1})}
\newcommand{\spset}[1]{\Sigma^n_{#1}}

\newcommand{\lrset}[1]{\cl C_{#1}^{n_1\times n_2}}

\newcommand{\Amap}{{\sf A}}

\newcommand{\Dmap}{{\sf D}}

\newcommand{\scp}[2]{\langle #1,\, #2 \rangle}

\newcommand{\Rbb}{\mathbb{R}}

\newcommand{\inv}[1]{\frac{1}{#1}}

\newcommand{\tinv}[1]{{\textstyle\frac{1}{#1}}}

\renewcommand{\leq}{\leqslant}
\renewcommand{\geq}{\geqslant}

\DeclareMathOperator{\ve}{vec}

\newcommand{\bs}{\boldsymbol}
\newcommand{\bb}{\mathbb}
\newcommand{\cl}{\mathcal}

\newcommand{\ts}{\textstyle}

\newcommand{\ie}{\emph{i.e.}, }
\newcommand{\eg}{\emph{e.g.}, }

\newcommand{\iid}{%
    \ifmmode
        \mathrm{i.i.d.}%
    \else%
        i.i.d.\@\xspace%
    \fi%
}
\newcommand{\rv}{\mbox{r.v.\@}\xspace}
\newcommand{\whp}{\mbox{w.h.p.\@}\xspace}
\newcommand{\st}{\mbox{s.t.\@}\xspace}

\newcommand{\myparagraph}[1]
{\ \\[-3mm]
\noindent\textbf{#1}\ }

\title{Taking the edge off quantization:\\
projected back projection in dithered compressive sensing\vspace{-2mm}}
\author{Chunlei Xu\thanks{Email: {\{chunlei.xu,vincent.schellekens,laurent.jacques\}@uclouvain.be}. 
    The authors are funded by the Belgian F.R.S.-FNRS. Part of this study is funded by the project {\sc AlterSense} (MIS-FNRS).},\ Vincent Schellekens$^*$ and Laurent Jacques$^*$\\[2mm]
  \small ISPGroup, ICTEAM/ELEN, UCLouvain, Belgium.}

\begin{document}

\maketitle
\begin{abstract}
Quantized compressive sensing (QCS) deals with the problem of representing compressive signal measurements with finite precision representation, \ie a mandatory process in any practical sensor design. To characterize the signal reconstruction quality in this framework, most of the existing theoretical analyses lie heavily on the quantization of sub-Gaussian random projections (\eg Gaussian or Bernoulli). We show here that a simple uniform scalar quantizer is compatible with a large class of random sensing matrices known to respect, with high probability, the restricted isometry property (RIP). Critically, this compatibility arises from the addition of a uniform random vector, or \emph{dithering}, to the linear signal observations before quantization. In this setting, we prove the existence of (at least) one signal reconstruction method, \ie the projected back projection (PBP), whose reconstruction error decays when the number of quantized measurements increases. This holds with high probability in the estimation of sparse signals and low-rank matrices. We validate numerically the predicted error decay as the number of measurements increases.\ \\[3mm]  

\end{abstract}

\section{Introduction}
\label{sec:intro}

To release the burden of high-dimensional signal sampling combined with
post-processing compression methods, Compressive Sensing (CS) theory
\cite{CT2005,FR2013} has emerged as a new procedure to compressively and
non-adaptively sample low-complexity signals, \eg sparse in a certain
basis or following a low-rank model. 

Specifically, CS shows how to recover a signal $\bs x$ that (approximately)
belongs to a low-complexity set $\cl K \subset \bb R^n$ from its
compressive measurement vector $\bs y = \bs \Phi \bs x + \bs n \in \bb
R^m$, where $\bs y$ is acquired from a sensing (or measurement) matrix $\bs \Phi \in \bb
R^{m\times n}$ (with $m< n$) with an additive noise $\bs n \in \bb R^m$. 
Many non-linear reconstruction algorithms then attain a stable and robust estimation of $\bs x \in \cl K$ from $\bs y$ by leveraging the low-complexity signal model 
(\eg $\ell_1$-norm minimization, greedy algorithms
\cite{CT2005,TG2007,FR2013}). The accuracy of this estimate can be ensured if $\inv{\sqrt{m}}\bs
\Phi$ respects the Restricted Isometry Property (RIP) over $\cl
K$, \ie $\|\inv{\sqrt{m}}\bs
\Phi \bs u\| \approx \|\bs u\| := (\sum_i |u_i|^2){}^{1/2}$ for all $\bs u \in \cl K$, up to a (multiplicative) distortion decreasing when $m$ increases. 
Since the advent of CS, numerous random matrix constructions (\eg the unstructured
sub-Gaussian random matrices or the structured random partial Fourier matrix) have been discovered to respect the RIP with high probability (\whp)~\cite{BDDW08,MPT2008,Rau10,FR2013}. 

As a matter of fact, actual acquisition systems cannot obtain
infinite precision on the recorded data. Signal observations must be digitized for transmission purposes,
storage or further specific processing.  Therefore, a more realistic Quantized CS
(QCS) model lies in estimating a low-complexity signal $\bs x\in
\cl K$ from $\bs y = \cl Q^{\rm g}(\bs \Phi \bs x)$,
where $\cl Q^{\rm g}: \bs u \in \bb R^m \mapsto \cl Q^{\rm g} (\bs u) \in \cl A \subset \bb R^m$ is a general \emph{quantization} function  (or \emph{quantizer}) mapping $m$-dimensional vectors to vectors in a discrete set (or \emph{codebook}) $\cl A$. Many quantizers have been studied in QCS, \eg $\Sigma\Delta$-quantization~\cite{gunturk2013sobolev,HS2018QCSStruMatrix}, non-regular scalar quantizers~\cite{B_TIT_12}, 
non-regular \emph{binned}
quantization~\cite{pai_nonadapt_MIT06,kamilov_2012}, and vector
quantization by frame permutation~\cite{vivekQuantFrame}. They aim
at easing the impact of the quantizer on the estimation of $\bs x$ from $\bs y$, by some well-designed algorithms achieving fast (\eg polynomial or exponential) reconstruction error decay when $m$ increases~\cite{gunturk2013sobolev,BFNPW2017,CSQuanMeasure2010}. These works are mostly dominated by the use of sub-Gaussian random matrices.
Only two recent studies escape from this domination:~\cite{DJR2017} uses partial circulant ensembles with Gaussian random entries in 1-bit CS, and~\cite{HS2018QCSStruMatrix} leverages fast Johnson-Lindenstrauss embeddings based on bounded orthonormal systems (BOS) and partial circulant ensembles (PCE) with noise-shaping quantization (\eg $\Sigma\Delta$). 

\myparagraph{Contributions:} The standpoint of this work is to show that a simple, non-optimal scalar quantization procedure, \ie a uniform quantizer $\cl Q(\cdot):=\delta \lfloor\frac{\cdot}{\delta} \rfloor$ (with $\lfloor \cdot \rfloor$ the floor function) of resolution $\delta>0$, applied componentwise onto vectors (or entry-wise on matrices), is \emph{compatible with a large class of sensing matrices known to satisfy the RIP}. This includes not only the unstructured sub-Gaussian random constructions, but also structured sensing matrices such as random partial Fourier/DCT matrices, BOS or PCE random constructions~\cite{FR2013}. This compatibility arises \emph{iff}\footnote{Actually, without dithering, there exist signals that cannot be estimated in QCS with Bernoulli sensing (see \eg~\cite[Sec. 5]{J2015}).} a random, uniform \emph{dithering} $\bs \xi\in\bb R^m$, with $\xi_i \sim_{\iid} \cl U([0,\delta])$, is added to the quantizer input~\cite{B_TIT_12,J2015,JC2016}, yielding the new sensing model:\vspace{-1mm}
\begin{equation}
\label{eq:Uniform-dithered-quantization}
\ts \bs y = \Amap(\bs x) = \Amap(\bs x; \bs\Phi, \bs\xi) :=\cl Q(\bs \Phi\bs x+\bs \xi) \in \delta \bb Z^m.\vspace{-1mm}
\end{equation}
Surprisingly, the announced compatibility between the QCS model \eqref{eq:Uniform-dithered-quantization} and the class of RIP matrices is actualized by a simple yet effective reconstruction method, the projected back projection (PBP) of the quantized measurements onto the set $\cl K$. This amounts to finding the closest point to the back projection (BP) $\inv{m}\bs \Phi^T\bs y$ in $\cl K$. 

Moreover, given a fixed sensing matrix satisfying the RIP, we show that PBP achieves \whp on the draw of the dithering good reconstruction performances in two cases: for the \emph{uniform estimation} of all signals in $\cl K$ given one draw of a random dithering, and for the \emph{non-uniform estimation} of one single signal with a dithering generated conditionally to this signal.

\myparagraph{Prerequisites and Assumptions:} To derive our results, we first assume that the signal set $\cl K$ is a \emph{structured low-complexity} (SLC) set $\cl K\subset \bb R^n$. Mathematically, this means that \emph{(i)} $\cl K \ni \bs 0$, \emph{(ii)} $\cl K$ is a \emph{cone}, \ie $\lambda \cl K \subset \cl K$ for all $\lambda > 0$, and \emph{(iii)} the Kolmogorov entropy of $\cl K \cap\bb B^n$ is bounded as\vspace{-2mm} 
\begin{equation}
  \label{eq:bound-kolm-ent}
  \ts \cl H(\cl K \cap \bb B^n, \eta) \leq \setcst{\cl K} \log(1 + 1/\eta),\vspace{-2mm}  
\end{equation}
where $\exp(\cl H(\cl S, \eta))$ is the smallest number of translated $\ell_2$-balls of radius $\eta>0$ that can \emph{cover} $\cl S\subset \bb R^n$, and $\setcst{\cl K}>0$ only depends on the geometry of $\cl K$. Good examples of SLC sets are $\Sigma^n_k$, the set of $k$-sparse signals in $\Rbb^n$ (with\footnote{Henceforth, the symbols $C,C',C'', \cdots, c,c',c'', \cdots > 0$ are positive and universal constants whose values can change from one line to the other.} $\setcst{\cl K} \leq C k \log(n/k)$), and $\lrset{r}$, the set of rank-$r$ $(n_1\!\times\!n_2)$-matrices (with $n_1n_2 = n$ and $\setcst{\cl K} \leq C r (n_1 + n_2)$). Note that $\setcst{\cl K}$ and the square Gaussian mean width (SGMW) of ${\cl K \cap \bb B^n}$, \ie another measure of a set dimension~\cite{KM2005,LM2013}, are not equivalent but often share the same bounds (see, \eg~\cite{oymak2015near,JC2016} for more examples). 

Additionally, for the analysis of the decay rate of PBP in Sec.~\ref{sec:lpd-A}, we consider that $\tinv{\sqrt m}\bs \Phi$ is generated by a \emph{random embedding of low-complexity set} (RELS) construction such that, given a distortion $\epsilon>0$, a failure probability $0<\zeta<1$, and the constant $\setcst{\cl K}>0$ defined above, if
$m \geq C \epsilon^{-2} \setcst{\cl K}\,\cl P_{\log}(m, n, 1/\zeta)$, where $\cl P_{\log}$ is some low-degree polynomial of logarithms in its arguments, then $\tinv{\sqrt m}\bs\Phi$ respects the RIP$(\cl K, \epsilon)$, \ie\vspace{-2mm}  
\begin{equation}
\label{eq:RIP-K-B}
\ts |\frac{1}{m}\|\bs \Phi\bs u\|^2-\|\bs u\|^2|\leq\epsilon,~\forall \bs u\in \cl K\cap\bb B^n,\vspace{-1mm}
\end{equation} 
with the probability exceeding $1-\zeta$. 

RELS constructions actually compose the vast majority of random matrix constructions known to satisfy the RIP~\cite{FR2013}.  This is the case of sub-Gaussian random matrices or Partial Random Orthonormal Matrix (PROM) over any SLC set $\cl K$ (with $\setcst{\cl K}$ bounding the SGMW of $\cl K$)~\cite{KM2005,oymak2015sors,XJ2018}, BOS or PCE over sparse signals, or other constructions listed in~\cite{XJ2018}. For instance, a random matrix $\tinv{\sqrt m}\bs \Phi$ generated by a (discrete) BOS is RIP$(\cl K, \epsilon)$ with probability exceeding $1-\zeta$ over the set of $k$-sparse signals in an orthonormal basis $\bs \Psi\in \bb R^{n\times n}$, \ie $\cl K= \bs \Psi\spset{k}$, provided $\ts m \geq C \mu^2 \epsilon^{-2} k \,(\log k)^2 \log n \log m \log 1/\zeta$, with $\mu>0$ the coherence of the BOS with $\bs \Psi$~\cite{CT06,Rau10,FR2013}. This matches the RELS requirement on $m$, \eg with the classical bound $\setcst{\cl K} \leq C k \log(n/k)$~\cite{BDDW08}.

\myparagraph{Paper organization:} The rest of the paper is structured as follows. First, we prove in Sec.~\ref{sec:DCS}, that PBP can actually deliver good estimates for signals in a certain SLC set $\cl K$ observed by the general \textit{distorted} CS (DCS) model, \vspace{-2mm}   
\begin{equation}
  \label{eq:distorted-CS-model}
 \ts  \bs y = \Dmap(\bs x) \in \bb R^m,\quad \bs x \in \cl K \cap \bb B^n, \vspace{-1mm}
\end{equation}
where $\bb B^n$ is the unit $\ell_2$-ball. This fact is ensured when
the distorted mapping $\Dmap: \bb R^n\to \bb R^m$, which includes the
dithered quantizer $\Amap$
in~\eqref{eq:Uniform-dithered-quantization}, respects a certain
\emph{limited projection distortion} (LPD) property that somehow qualifies
how far $\Dmap$ is from a linear mapping. Next, in Sec.\ref{sec:lpd-A}, we establish that, \whp, the reconstruction error of PBP decays like $O(m^{-1/2})$, up to log factors, for the set of sparse signals and the set of low-rank matrices and in the context of quantized RELS observations. Finally, in Sec.~\ref{sec:exp}, we validate our results numerically in various experiments involving different SLC sets, sensing matrices, and under multiple sensing parameters.

\myparagraph{Related works:} Reconstruction of low-complexity signals from QCS observations has been studied in the context of 1-bit CS~\cite{BB2008,BFNPW2017,PV2013,DJR2017} and multi-bit quantization~\cite{JDV2015,OptQuanLass2016,HS2018QCSStruMatrix}. Most of these works focus on estimating such signals from their quantized or non-linearly disturbed sub-Gaussian random projections. The studies~\cite{HS2018QCSStruMatrix} and~\cite{DJR2017} are two exceptions that use, respectively, BOS and PCE constructions, and subsampled Gaussian random circulant sensing matrix. However, both works are restricted to sparse signal estimations. Variants of the LPD property defined in Sec.~\ref{sec:DCS} were introduced in~\cite[Thm 1.9]{PV2016} and in~\cite{JDV2015,PV2013} for bounding signal reconstruction error in non-linear CS and in 1-bit CS, respectively.  Adaptive or random dithering were also considered in 1-bit CS~\cite{BFNPW2017,DJR2017} and in multi-bit QCS~\cite{DJR2017}. 
Finally, by instantiating the non-linear CS models of~\cite{PVY2017,PV2016} to the QCS model~\eqref{eq:Uniform-dithered-quantization}, our results are essentially recovered in the specific case of non-uniform sparse signal estimation with quantized, dithered Gaussian random projections. In this sense, our work can thus be seen as a generalization of this context to quantized, dithered random projections of signals with RIP matrices, involving both more general low-complexity signal sets and uniform reconstruction guarantees. 

\section{PBP reconstruction error in DCS}
\label{sec:DCS}

The PBP estimate of a signal $\bs x\in \cl K$ observed by the DCS model $\bs y = \Dmap(\bs x)$ is formally defined as\vspace{-2mm}
\begin{equation}
  \label{eq:PBP-estimate}
\ts  \hat{\bs x}\ :=\ \cl P_{\cl K}(\inv{m} \bs \Phi^\top \bs y),\vspace{-2mm} 
\end{equation}
where $\inv{m} \bs \Phi^\top \bs y$ stands for the back projection of the measurement vector $\bs y$, and $\cl P_{\cl K}$ is a projector\footnote{In cases where $\min_{\bs u \in \cl K} \|\bs z - \bs u\|$ has several minimizers, \eg if $\cl K$ is non-convex, $\cl P_{\cl K}$ picks one of them arbitrarily.} on $\cl K$, \ie
$
\ts \cl P_{\cl K}(\bs z)\ \in\ \arg\min_{\bs u \in \cl K} \|\bs z - \bs u\|.
$
Throughout this work, we assume $\cl P_{\cl K}$ can be computed in polynomial time with respect to $m$ and $n$. 
For instance, if $\cl K$ is the set of $k$-sparse vectors or the set of rank-$r$ matrices, $\cl P_{\cl K}$ is the hard thresholding operator zeroing all but the $k$ greatest in absolute value components of vectors, or zeroing all but the $k$ first singular values of matrices in their SVD decomposition.

PBP can provide accurate estimate of a low-complexity signal $\bs x \in \cl K$ observed by the DCS model~\eqref{eq:distorted-CS-model} if the mapping $\Dmap$ is not too far from a RIP matrix $\tinv{\sqrt m}\bs \Phi \in \bb R^{m \times n}$. Mathematically, given a linear mapping $\bs \Phi$ and a distortion $\nu > 0$, this amounts to asking $\Dmap$ to respect the limited projected distortion (LPD) property over a set $\cl K\subset \bb R^n$ observed by $\bs \Phi$, or LPD$(\cl K, \bs \Phi, \nu)$, which reads 
\begin{equation}\label{eq:LPD}
\ts \inv{m}\,|\langle \Dmap(\bs u), \bs \Phi \bs v\rangle-\langle \bs
\Phi\bs u,\bs \Phi\bs v\rangle|\ \leq \nu,\quad \forall \bs u, \bs v \in \cl K \cap \bb B^n.
\end{equation}
This property can be \emph{localized} if $\bs u$ is fixed in~\eqref{eq:LPD}, in which case 
$\Dmap$ respects the \emph{local} LPD property on $\bs u$, or L-LPD$(\cl K, \bs \Phi, \bs u, \nu)$.

For fixed $\bs u,\bs v\in \cl K\cap \bb B^n$, the (L)LPD property bounds the scalar product between the deviation $\Dmap(\bs u)-\bs \Phi\bs u$ and undistorted compressed observations $\bs \Phi \bs v$ in the compressed domain $\bb R^m$. If the distortion is solely an additive noise, \ie $\Dmap(\bs u) = \bs \Phi \bs u + \bs \rho$, proving the L-LPD degenerates to showing that $\tinv{m}\scp{\bs \rho}{\bs \Phi \bs v}$ is small for any $\bs v\in \cl K \cap \bb B^n$ and a fixed $\bs \rho$. This is easy to prove when the components of $\bs \rho$ are \iid sub-Gaussian, which includes the QCS model~\eqref{eq:Uniform-dithered-quantization} as every \iid \rv $\rho_i := \cl Q((\bs \Phi \bs u)_i + \xi_i)-(\bs \Phi \bs u)_i$ is bounded and thus sub-Gaussian (see Sec.~5 in~\cite{XJ2018} for the proof). However, this cannot be directly generalized to the uniform LPD property (meaning that~\eqref{eq:LPD} would hold for all $\bs \rho$) without considering the geometry of $\Dmap$. In the case where $\Dmap \equiv \Amap$, we shall in particular control the impact of discontinuities introduced by $\cl Q$ on $\bs \rho$ to prove that the LPD holds under certain conditions (see Sec.~\ref{sec:lpd-A}).       

As detailed below, it is easy to understand why PBP can provide good signal estimate. We note first that a standard use of the polarization identity proves that if $\inv{\sqrt{m}}\bs \Phi$ satisfies the RIP$(\cl K-\cl K,\epsilon)$, then $\ts \inv{m}\,|\langle \bs
\Phi\bs u,\bs \Phi\bs v \rangle -\langle \bs u,\bs v\rangle|\leq 2\epsilon,~\forall \bs u, \bs v \in \cl K \cap \bb B^n$ (see, \eg~\cite{FR2013},~\cite[Lemma 3.5]{XJ2018}). Therefore, under the LPD$(\cl K, \bs \Phi, \nu)$ of $\Dmap$, the triangular identity provides
$$\ts \,|\inv{m} \langle \Dmap(\bs u), \bs \Phi \bs v\rangle- \langle
  \bs u,\bs v\rangle|\ \leq\ 2\epsilon + \nu,\quad \forall \bs u, \bs v \in \cl K \cap \bb B^n.
$$
Consequently, if $\bs y$ is the DCS observation of $\bs x \in \cl K$, maximizing $\scp{\bs v}{\tinv{m}\bs\Phi^\top \bs y}$ with some $\bs v \in \cl K$, as done somehow by $\cl P_{\cl K}$~in~\eqref{eq:PBP-estimate}, is a good proxy for maximizing the correlation of $\bs v$ with $\bs x$, \ie the optimal $\bs v$ is \st $\bs v \approx \bs x$. Here is a more rigorous explanation.

\begin{Thm}[PBP error on sparse signals]
\label{thm:PBP-sparse-vector}
Given two distortions $\epsilon, \nu > 0$, if $\tinv{\sqrt m}\bs \Phi \in \bb R^{m \times n}$ respects the RIP$(\spset{2k}, \epsilon)$ and if the mapping $\Dmap$ satisfies the LPD$(\spset{2k}, \bs \Phi, \nu)$, then, 
for all $\bs x \in \spset{k} \cap \bb B^n$, the estimate $\hat{\bs x}$ obtained by the PBP of $\bs y=\Dmap(\bs x)$ onto $\spset{k}$ satisfies 
\begin{equation}
  \label{eq:PBP-sparse-vector}
 \ts \|\bs x-\hat{\bs x}\| \leq 4\epsilon+2\nu,\quad \forall \bs x \in \cl K \cap \bb B^n.\vspace{-1mm}  
\end{equation}
If $\bs x$ is fixed, \eqref{eq:PBP-sparse-vector} holds if $\Dmap$ respects the L-LPD$(\spset{2k}, \bs \Phi, \bs x, \nu)$.
\end{Thm}
\begin{proof}
Denote $T\subset [n]$ as the union of the supports of $\bs x$ and $\hat{\bs x}$, thus $|T|\leq 2k$, and let $\bs a :=\frac{1}{m}\bs \Phi^T\bs y$. Since $\hat{\bs x}= \cl P_{\cl K}(\bs a)$, $\hat{\bs x}$ is also the best $k$-term approximation of $\bs a_T$ zeroing all but the entries of $\bs a$ indexed in $T$. Therefore,
$\ts \|\bs x-\hat{\bs x}\|\leq\|\bs x-\bs a_T\|+\|\bs a_T-\hat{\bs x}\|\leq 2\|\bs x-\bs a_T\|$. Since $\bs \Phi$ and $\Dmap$ respect the RIP$(\spset{2k},\epsilon)$ and the LPD$(\spset{2k}, \bs \Phi,\nu)$, respectively, we have\vspace{-2mm}  
\begin{multline*}
\ts \|\bs x-\hat{\bs x}\| \leq2\|\bs x-\bs a_T\|=\ts 2\sup_{\bs w\in \bb B^n}\langle \bs w, \bs x-\bs a_T\rangle\\[-1mm]
\ts =2\sup_{\bs w\in \spset{T}\cap\bb B^n} [\langle \bs w, \bs x\rangle-\inv{m}\langle\bs \Phi\bs w,\Dmap(\bs x)\rangle]\leq 4\epsilon+2\nu,\\[-7mm]
\end{multline*}
where $\spset{T}$ is the set of vectors in $\bb R^n$ supported on $T$.
Moreover, if $\bs x$ is fixed, we clearly see that only the L-LPD$(\spset{2k}, \bs \Phi, \bs x, \nu)$ is required, which completes the proof.
\end{proof}

Up to a vectorization\footnote{$\bs x=\ve(\bs X)$ stacks all the columns of $\bs X$ on top of one another.} of the matrix domain $\bb R^{n_1\times n_2}$, \ie identifying this space with $\bb R^n$ (with $n=n_1n_2$) and allowing for the DCS observation of matrices in~\eqref{eq:distorted-CS-model}, we can proceed similarly to bound the reconstruction error of PBP in the estimation of low-rank matrices. The proof is similar to the one of Theorem~\ref{thm:PBP-sparse-vector} once we identify a common subspace for both the observed rank-$r$ matrix $\bs X\in \lrset{r}\cap \bb B^{n_1\times n_2}_F$ and its PBP estimate $\hat{\bs X}$, where $\cl K=\lrset{r}:=\{\bs Z\in \bb R^{n_1\times n_2}:~\text{rank}(\bs Z)\leq r\}$ and $\bb B^{n_1\times n_2}_F := \{\bs Z \in \bb R^{n_1\times n_2}: \|\bs Z\|_F := \|\!\ve(\bs Z)\| \leq 1\}$ is the Frobenius unit ball. As a result, the reconstruction error of PBP is bounded by $\|\bs X-\hat{\bs X}\|_F \leq 4\epsilon+2\nu$, provided that $\tinv{\sqrt m}\bs \Phi$ and $\Dmap$ respect the RIP$(\lrset{4r}, \epsilon)$ and the LPD$(\lrset{4r}, \bs\Phi,\nu)$, respectively (see.~\cite[Theorem 4.2]{XJ2018}). In Sec.~\ref{sec:exp}, we numerically validate the error distortions of PBP over both $\spset{k}$ and $\lrset{r}$.

\section{Error Decay Analysis of PBP}
\label{sec:lpd-A}

In this section, we establish how the reconstruction error of PBP decays when $m$ increases. This is done in the particular case where sparse signals or low-rank matrices are observed from the QCS model~\eqref{eq:Uniform-dithered-quantization} endowed with a random uniform dithering and a matrix $\bs \Phi$ generated from a RELS construction (see Sec.~\ref{sec:intro}). 

Since this study is supported by the general results of the previous section, we need first to determine when the quantized mapping $\Amap$ generated from a RIP matrix respects the LPD \whp on the drawn of the dithering.  We go thus beyond the L-LPD property, which trivially holds for the mapping $\Amap$ (see Sec.~\ref{sec:DCS}), by cautiously analyzing the interplay between the quantizer discontinuities and the dithering.   
\begin{Propn}[LPD for $\Amap$ over SLC set]
\label{prop:main-result-embed}
Given a SLC set $\cl K \subset
\bb R^n$, a distortion $0<\epsilon <1$, a quantization resolution $\delta > 0$, a matrix $\tinv{\sqrt m} \bs \Phi$ respecting the RIP$(\cl K-\cl K,\epsilon)$\footnote{$\cl K-\cl K$ denotes the Minkowski difference of $\cl K$ with itself.}, a random dithering $\bs \xi \sim \cl U^m([0,\delta])$, and provided
the random mapping $\Amap$ in~\eqref{eq:Uniform-dithered-quantization} respects 
the LPD$(\cl K, \bs \Phi, \epsilon(1+\delta))$
with probability exceeding $1 - C' \exp(- c' \epsilon^2 m)$.
\end{Propn}
The full proof of this proposition is given in~\cite[Sec.~6]{XJ2018}. We provide here an intuitive proof sketch pruned of too technical considerations. First, for a fixed pair of vectors $\bs \Phi\bs u,\bs \Phi\bs v\in \bb R^n$, notice that $\bb E_d \lfloor \lambda + d\rfloor = \lambda$, for $\lambda \in \bb R$ and $d\sim \cl U([0,1])$, induces $\bb E_{\bs \xi} \langle \Amap(\bs u), \bs \Phi \bs v\rangle-\langle \bs \Phi\bs u,\bs \Phi\bs v\rangle = 0$ (see~\cite[Lem. A.1]{XJ2018}). Since, asymptotically in $m$, $\langle \Amap(\bs u), \bs \Phi \bs v\rangle$ approaches $\bb E_{\bs \xi} \langle \Amap(\bs u), \bs \Phi \bs v\rangle$, $\langle \Amap(\bs u) - \bs \Phi \bs u, \bs \Phi \bs v\rangle$ should thus tend to 0, as targeted by the LPD. In fact, using measure concentration tools on the sub-Gaussianity of $\Amap(\bs u) - \bs \Phi \bs u$, we can show that, with probability exceeding $1-2\exp(-2\epsilon^2m)$, $R(\bs u, \bs v):=|\langle \Amap(\bs u) - \bs \Phi \bs u, \bs \Phi \bs v\rangle| \leq \delta\epsilon \sqrt{m}\|\bs \Phi\bs v\|$ (see~\cite[Lem. 6.3]{XJ2018}). Moreover, since $\cl K$ is a SLC set and $\bs 0\in \cl K$, the RIP$(\cl K-\cl K,\epsilon)$ defined in~\eqref{eq:RIP-K-B} involves that $\|\bs \Phi\bs v\|\leq \sqrt{m(1+\epsilon)} \|\cl K\cap \bb B^n\|\leq\sqrt{2m}$, so that $R(\bs u, \bs v) \leq \sqrt 2 \delta \epsilon m$ with the same probability.

Second, we must bound $R(\bs u, \bs v)$ for all vectors $\bs u, \bs v \in \cl K \cap \bb B^n$. In the case where $\bs u$ is fixed, we can bound $R$ for all $\bs v \in \cl K \cap \bb B^n$ by a standard covering-and-continuity argument~\cite{BDDW08}.  In other words, if $\cl K_\eta \subset \cl K \cap \bb B^n$ is an optimal $\eta$-covering of $\cl K \cap \bb B^n$, \ie $\cl K \subset \cup_{\bs q \in \cl K_\eta} \{\bs q + \eta \bb B^n\}$ with $\log |K_\eta| = \cl H(\cl K, \eta)$ (with $\cl H$ the Kolmogorov entropy introduced in Sec.~\ref{sec:intro}), then a union bound provides that, for all $\bs q \in \cl K_\eta$, $R(\bs u, \bs q) \leq \sqrt 2 \delta \epsilon m$ with probability exceeding $1-2\exp(\cl H(\cl K \cap \bb B^n, \eta) -2\epsilon^2m)$. Since any $\bs v \in \cl K \cap \bb B^n$ is associated to an $\eta$-close element of $\cl K_\eta$, this last result can basically be extended with the same probability to all $\bs v \in \cl K \cap \bb B^n$ from the continuity of the scalar product, and by adequately connecting $\eta$ to~$\epsilon$. 

However, a similar treatment cannot be applied for an extension to all $\bs u \in \cl K \cap \bb B^n$ since the quantizer discontinuities in $\Amap$ prevent directly using the same continuity argument.
We can fortunately overcome this issue by showing that, for all $\cl V$ picked in the covering neighborhoods $\bb V := \{\bs q + \eta \bb B^n: \bs q \in \cl K_\eta\}$, the number of components of $\Amap$ being discontinuous over $\cl V$ constitutes, \whp, only a small fraction of $m$. Therefore, $R(\bs u, \bs v)$ can be bounded for all $\bs u \in \cl K \cap \bb B^n$ by: \emph{(i)} bounding it, by union bound, over all elements of $\cl K_\eta$, and \emph{(ii)}, for all $\bs u \in \cl K \cap \bb B^n$, splitting the separable scalar product in $R$ into two parts, one composed of all continuous components of $\Amap$ over the neighborhood of $\bb V$ containing $\bs u$, and which can then be bounded by continuity, and the other composed of a minority of discontinuous components bounded by using the crude deterministic bound $|(\Amap(\bs u) - \bs \Phi \bs u)_i| \leq 2\delta$. Gathering all these bounds, and adjusting $\eta$ to $\epsilon$ (\ie $\eta = c\epsilon^3$), then provides~\eqref{eq:LPD} with $\nu=\epsilon(1+\delta)$, and completes the proof.\\[-3mm]

We can now focus on the main result of this section, \ie determining the reconstruction error decay of PBP for the estimation of signals and matrices in $\spset{k}\cap \bb B^n$ or $\lrset{r}\cap \bb B_F^{n_1\times n_2}$, respectively, when they are observed from~\eqref{eq:Uniform-dithered-quantization}. We assume $\Amap$ endowed with a random uniform dithering, and $\tinv{\sqrt m}\bs \Phi$ generated from a RELS construction. Below, guided by the requirements of Thm.~\ref{thm:PBP-sparse-vector} and its extension to low-rank matrix estimation, the SLC set $\cl K'$ denotes either the set $\spset{2k}$ in the case of $k$-sparse signal estimation, or the set $\lrset{4r}$ for rank-$r$ matrix estimation.  

We follow the recommendations given in Sec.~\ref{sec:DCS} and the requirements imposed by Prop.~\ref{prop:main-result-embed}.  By the definition of RELS (Sec.~\ref{sec:intro}), if $m \geq \epsilon^{-2} \setcst{\cl K'}\cl P_{\log}(m, n, 1/\zeta)$, $\frac{1}{\sqrt{m}} \bs \Phi$ respects the RIP$(\cl K',\epsilon)$ with probability exceeding $1-\zeta$. Moreover, since $\cl K'$ is a SLC set whose Kolmogorov entropy is bounded as in~\eqref{eq:bound-kolm-ent}, the requirement on $m$ in Prop.~\ref{prop:main-result-embed} holds if $m \geq C \epsilon^{-2} \setcst{\cl K'} \log(1 + c\epsilon^{-3})$. Under this condition, the considered quantized mapping $\Amap$ satisfies thus the LPD$(\cl K', \bs \Phi, \epsilon(1+\delta))$ with probability exceeding $1-C\exp(-c\epsilon^2 m)$. 

Hence, by union bound over the events above, a few manipulations show that provided\vspace{-1mm} 
\begin{equation}
\label{eq:m-lpd-rip}
\ts m\geq \epsilon^{-2} \setcst{\cl K'}\,\cl P_{\log}(m, n, 1/\zeta, 1/\epsilon^3),\vspace{-1mm}
\end{equation}
the LPD$(\cl K', \bs \Phi,\epsilon(1+\delta))$ and the RIP$(\cl K',\epsilon)$ properties of $\tinv{\sqrt m}\bs \Phi$ and $\Amap$, respectively, both hold with probability exceeding $1-2\zeta$.  This finally guarantees $\|\bs x-\hat{\bs x}\|\leq C \epsilon(1+\delta)$ for all $\bs x \in \spset{k}\cap \bb B^n$, and equivalently for $\lrset{r}\cap \bb B_F^{n_1\times n_2}$ up to a vectorization.  

Equivalently, saturating the condition on $m$ in~\eqref{eq:m-lpd-rip} and inverting this relation with respect to $\epsilon$ provides $\epsilon = O((\setcst{\cl K'}/m)^{1/2})$, up to missing log factors. We can finally conclude this section and say that, within the precise context described above, uniformly or non-uniformly over the generation of $\bs \xi$, PBP provides, \whp, sparse signal or low-rank matrix estimates whose reconstruction error decays like
\begin{equation}
  \label{eq:error-decay-lsre-uls-low-rank}
\ts \|\bs x - \hat{\bs x}\| = O\big(\frac{1 + \delta}{\sqrt m}\, \setcst{\cl K'}^{-1/2}\big),
\end{equation}
when $m$ increases (up to missing log factors).

\section{Experiment results}
\label{sec:exp}

Let us now illustrate the evolution of the PBP reconstruction error when $m$ or $\delta$ increases. We do this for ``signals'' of $\spset{k}$ and $\lrset{r}$, for Gaussian and partial DCT random matrices (with DCT rows sampled without replacement), with and without dithering, and by carefully selecting our figures to avoid duplicated messages. 

\begin{figure}[htbp] 
   \begin{minipage}[b]{.48\textwidth}
        \centering   
\subfloat[\label{fig:gauss-sensing-sparse}4-sparse signals]
{\includegraphics{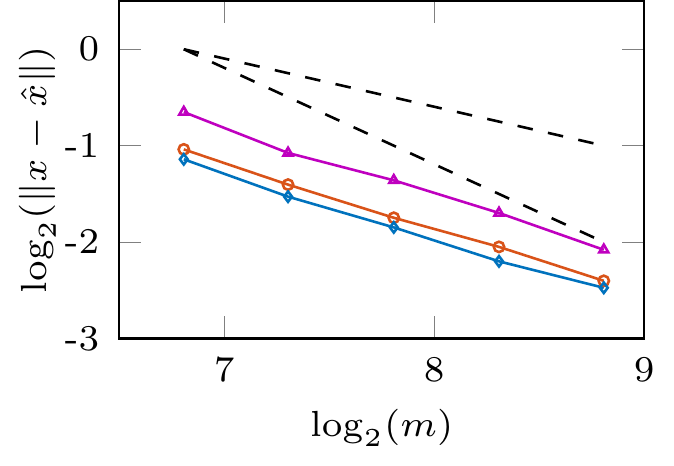}}
    \end{minipage}
    \hspace{0.5mm}
    \begin{minipage}[b]{.48\textwidth}
        \centering

 \subfloat[\label{fig:DCT-sensing-LR2}  rank-2 matrices]  
{\includegraphics{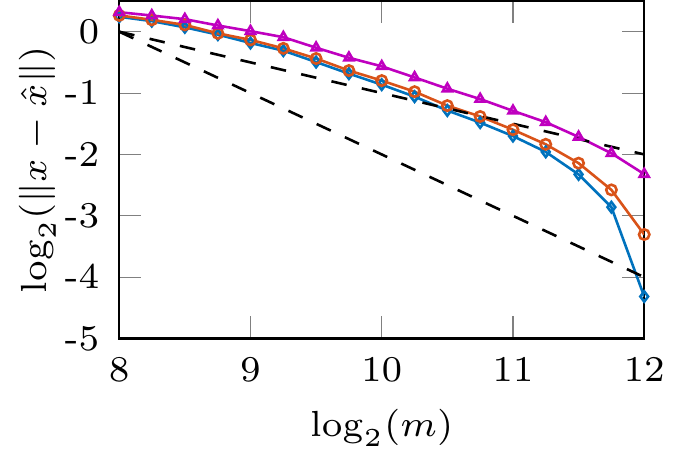}}
    \end{minipage}\vspace{-2mm}
 \caption{PBP reconstruction error evolution with $m$ (log-log plot), for $\delta = 0.5$ (blue diamonds), $\delta=1$
 (orange circles) and $\delta=2$ (pink triangles). Dashed lines indicate the rates $m^{-1/2}$ and $m^{-1}$.\vspace{-4mm}
}
  \label{fig:various-mtx-sensing}
 \end{figure}

\myparagraph{A.~Performances for two low-complexity sets:}
This experiment tests the relationship between the PBP reconstruction error of low-complexity signals and the number of measurements, for different quantization resolution $\delta$, where $\bs \Phi$ is either a Gaussian random matrix with elements drawn \iid from the standard normal distribution, or a partial DCT random matrix obtained by picking $m$ rows uniformly at random from an $n\times n$ orthonormal DCT matrix. 
 
For $\spset{k}$, we choose $n=512$, $k=4$ and\footnote{We expect $m=4k\log n/k$ (unquantized) linear observations suffice to reconstruct $k$-sparse signals.} $m \in [4k\log n/k, n]$. The signal $\bs x\in \spset{k}$ is obtained by picking one support uniformly at random amongst $n \choose k$ number of $k$-length supports of $[n]$, then drawing every $x_i$ in the support \iid from a standard normal distribution. Fig.~\ref{fig:gauss-sensing-sparse} shows the reconstruction error of
PBP of $4$-sparse signals as a function of $m$ for $\delta \in \{0.5,1,2\}$ displayed by three curves.  For every $\delta$ and $m$, the PBP reconstruction was tested over 100 trials of a random generation of $\bs \Phi$, $\bs \xi$ and $\bs x$.
We observe a reconstruction error decay rate slightly faster than $O(m^{-1/2})$ (\eg the curve at $\delta=1$ is well fitted by $O(m^{-0.67})$), as predicted by~\eqref{eq:error-decay-lsre-uls-low-rank}. 

Duplicating the experiment for $\lrset{r}$ with $n_1 = n_2 = 64$, $n=n_1n_2=4096$ and $r=2$, and inserting a partial DCT random matrix in $\Amap$, we can also show that the PBP reconstruction error decays as $m$ increases. Each rank-2 matrix was generated as $\bs X = c\bs B \bs C^\top$ with random matrices $\bs B, \bs C \in \bb R^{\sqrt n \times 2}$ having standard normal \iid entries, while $c>0$ ensures that $\|\bs X\|_F = 1$. The sensing matrix $\tinv{\sqrt m}\bs\Phi$ is a partial DCT random matrix operating over the vectorized form $\bs x = \ve(\bs X)$.  Fig.~\ref{fig:DCT-sensing-LR2} shows the decay of the reconstruction error of the PBP estimate $\hat{\bs x} = \ve(\hat{\bs X})$ when $m \in [n/16, n]$ increases (in a log-log plot) for $\delta \in \{0.5, 1, 2\}$ and with an average over 50 trials for each curve points (over the generation of $\bs \Phi$, $\bs \xi$ and $\bs X$). Specifically, as $m$ increases, the rate of the reconstruction error decay of PBP is faster than $O(m^{-1/2})$ for partial DCT random matrices. Another experiment, not presented here, over signals of $\spset{4}$ and with partial DCT random matrices also results in similar error decay.
 \begin{figure}[htbp] 
   \begin{minipage}[b]{.48\textwidth}
        \centering
\subfloat[\label{fig:DCT-sensing-sparse-nodith} 4-sparse signals]
{\includegraphics{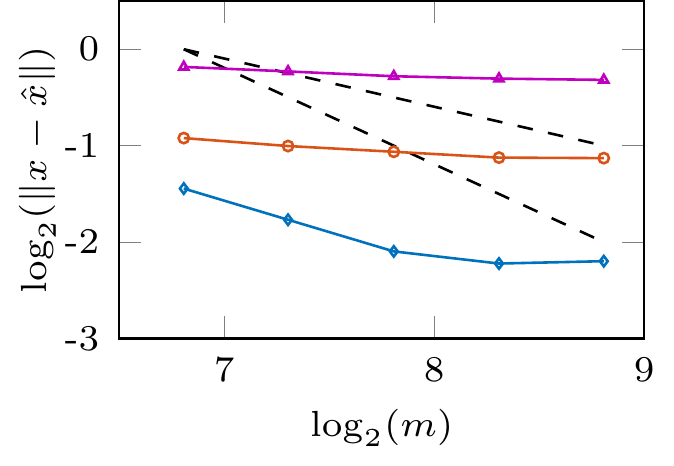}}
    \end{minipage}
    \hspace{1mm}
    \begin{minipage}[b]{.48\textwidth}
       \centering
 \subfloat[\label{fig:evol-gauss-sensing-delta-LR2}  rank-2 matrices ]
{\includegraphics{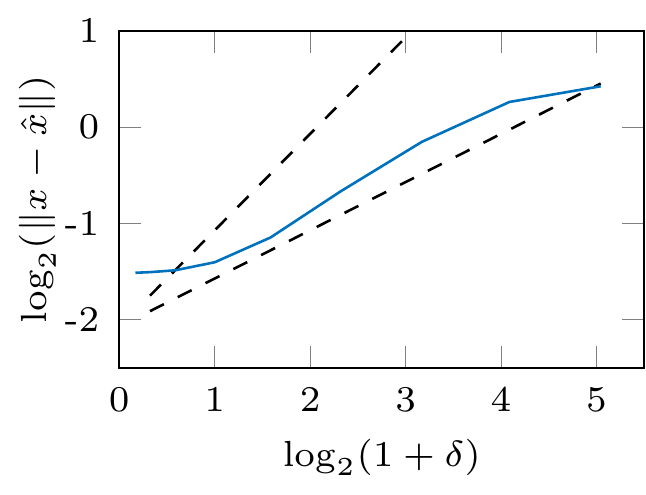}}
     \end{minipage}
 \caption{(a) PBP reconstruction error evolution with $m$ (log-log plot) for QCS observations (without dithering). (b) PBP reconstruction error evolution with~$\delta$ (log-log plot) from the QCS observations of matrices in $\lrset{2}$. Dashed lines indicate the rates $1/2\log_2(1+\delta)$ and $\log_2(1+\delta)$.\vspace{-4mm}}
  \label{fig:evol-gauss-sensing-delta}
  \end{figure}

\myparagraph{B.~Impact of the dithering:} We now generate signals in $\spset{4}$ as in the first experiment. These are then observed by configuring $\Amap$ with a partial DCT random matrix and a canceled dithering. Fig.~\ref{fig:DCT-sensing-sparse-nodith} demonstrates that the decay of the PBP reconstruction error reaches a constant floor when $m$ increases, especially at $\delta=2$. A similar phenomenon, not reported here, is also observed for the reconstruction error of rank-$2$ matrices. This confirms the positive impact of the dithering in the quantization, \ie it accelerates the decay rate of the reconstruction error of PBP. 

\myparagraph{C.~Impact of the quantization resolution:}
We finally evaluate the PBP reconstruction error on $\lrset{2}$ as a function of $\delta$ and for Gaussian random matrix. We set $\log_2 \delta \in [-3, 5]$ and kept $m=n/2$ fixed. In Fig.~\ref{fig:evol-gauss-sensing-delta-LR2}, we observe that the error curve is compatible with the theoretical (upper) bound $\|\bs x -\hat{\bs x}\| = O(1+\delta)$. The decay seems actually closer to $C \sqrt{1+\delta}$, which could be induced by the Gaussianity of the sensing. At small value of $\delta$, the error saturates to a floor, \ie the quantizer $\cl Q$ reduces to the identity operator when $\delta$ tends to zero. Repeating the experiment for $\spset{4}$, a similar error decay rate is observed when $\delta$ decreases (not reported here).
 
\section{Conclusion}
\label{sec:con}

Our work has demonstrated the existence of (at least) one reconstruction method, the \textit{projected back projection} (PBP), that reconciles RIP random matrices with the specific QCS model~\eqref{eq:Uniform-dithered-quantization} induced by a uniform scalar quantization. Critically, this reconciliation is possible from the addition of a uniform random dithering before quantizing the linear signal observations. Thanks to it, PBP is proved to achieve accurate estimations of signals belonging to SLC sets (\eg $\spset{k}$ and $\lrset{r}$), and this is confirmed numerically. Moreover, in the absence of dithering, we have also isolated numerical examples where the reconstruction performance saturates. Our numerical tests also confirm a general decay rate in $O(1/\sqrt{m})$ for the PBP reconstruction error of the considered signals as $m$ increases, up to missing factors.  As future works, we plan to extend this PBP study to other reconstruction algorithms, \eg using the PBP estimate as an initialization~\cite{JDV2015}. In particular, consistency between the signal estimate and the observed signal could accelerate the error decay, reaching the theoretic rate of $O(1/m)$ established in 1-bit CS and in QCS with Gaussian random matrix~\cite{JLBB2013,J2015}.


\footnotesize
\begin{thebibliography}{10}

\bibitem{CT2005}
E.~J. Cand\`es and T. Tao,
\newblock ``{Decoding by linear programming},''
\newblock {\em IEEE transactions on information theory}, vol. 51, no. 12, pp.
  4203--4215, 2005.

\bibitem{FR2013}
S. Foucart and H. Rauhut,
\newblock {\em {A mathematical introduction to compressive sensing}}, vol.~1,
\newblock Birkh{\"a}user Basel, 2013.

\bibitem{TG2007}
J.~A. Tropp and A.~C. Gilbert,
\newblock ``{Signal recovery from random measurements via orthogonal matching
  pursuit},''
\newblock {\em IEEE Transactions on information theory}, vol. 53, no. 12, pp.
  4655--4666, 2007.

\bibitem{BDDW08}
R. Baraniuk, M. Davenport, R. DeVore, and M. Wakin,
\newblock ``{A simple proof of the restricted isometry property for random
  matrices},''
\newblock {\em Constructive Approximation}, vol. 28, no. 3, pp. 253--263, 2008.

\bibitem{MPT2008}
S. Mendelson, A. Pajor, and N. Tomczak-Jaegermann,
\newblock ``{Uniform uncertainty principle for Bernoulli and subgaussian
  ensembles},''
\newblock {\em Constructive Approximation}, vol. 28, no. 3, pp. 277--289, 2008.

\bibitem{Rau10}
H. Rauhut,
\newblock ``{Compressive sensing and structured random matrices},''
\newblock {\em Theoretical foundations and numerical methods for sparse
  recovery}, vol. 9, pp. 1--92, 2010.

\bibitem{gunturk2013sobolev}
C.~S. G{\"u}nt{\"u}rk, M. Lammers, A.~M. Powell, R. Saab, and
  {\"O.}~Y{\i}lmaz,
\newblock ``{Sobolev duals for random frames and $\Sigma$$\Delta$ quantization
  of compressed sensing measurements},''
\newblock {\em Foundations of Computational mathematics}, vol. 13, no. 1, pp.
  1--36, 2013.

\bibitem{HS2018QCSStruMatrix}
T. Huynh and R. Saab,
\newblock ``Fast binary embeddings, and quantized compressive sensing with
  structured matrices,''
\newblock {\em arXiv preprint arXiv:1801.08639}, 2018.

\bibitem{B_TIT_12}
P.~T. Boufounos,
\newblock ``{Universal rate-efficient scalar quantization},''
\newblock {\em IEEE transactions on information theory}, vol. 58, no. 3, pp.
  1861--1872, 2012.

\bibitem{pai_nonadapt_MIT06}
R.~J. Pai,
\newblock {\em {Nonadaptive lossy encoding of sparse signals}},
\newblock Ph.D. thesis, Massachusetts Institute of Technology, 2006.

\bibitem{kamilov_2012}
U.~S. Kamilov, V.~K. Goyal, and S. Rangan,
\newblock ``{Message-passing de-quantization with applications to compressed
  sensing},''
\newblock {\em IEEE Transactions on Signal Processing}, vol. 60, no. 12, pp.
  6270--6281, 2012.

\bibitem{vivekQuantFrame}
H.~Q. Nguyen, V.~K. Goyal, and L.~R. Varshney,
\newblock ``{Frame permutation quantization},''
\newblock {\em Applied and Computational Harmonic Analysis}, vol. 31, no. 1,
  pp. 74--97, 2011.

\bibitem{BFNPW2017}
R.~G. Baraniuk, S. Foucart, D. Needell, Y. Plan, and M.
  Wootters,
\newblock ``{Exponential decay of reconstruction error from binary measurements
  of sparse signals},''
\newblock {\em IEEE Transactions on Information Theory}, vol. 63, no. 6, pp.
  3368--3385, 2017.
  
\bibitem{CSQuanMeasure2010}
A, Zymnis, S. Boyd and E.~J. Cand\`es,
\newblock ``{Compressive sensing with quantized measurements},''
\newblock {\em IEEE Signal Processing Letters}, vol. 17, no. 2, pp.
  149--152, 2010.  

\bibitem{DJR2017}
S. Dirksen, H.~C. Jung, and H. Rauhut,
\newblock ``{One-bit compressed sensing with partial Gaussian circulant
  matrices},''
\newblock {\em arXiv preprint arXiv:1710.03287}, 2017.

\bibitem{J2015}
L. Jacques,
\newblock ``{Small width, low distortions: quantized random embeddings of
  low-complexity sets},''
\newblock {\em IEEE Transactions on information theory}, vol. 63, no. 9, pp.
  5477--5495, 2015.

\bibitem{JC2016}
L. Jacques and V. Cambareri,
\newblock ``{Time for dithering: fast and quantized random embeddings via the
  restricted isometry property},''
\newblock {\em Information and Inference: A Journal of the IMA}, p. iax004,
  2017.

\bibitem{oymak2015near}
S. Oymak and B. Recht,
\newblock ``{Near-optimal bounds for binary embeddings of arbitrary sets},''
\newblock {\em arXiv preprint arXiv:1512.04433}, 2015.

\bibitem{KM2005}
B.~Klartag and S. Mendelson,
\newblock ``{Empirical processes and random projections},''
\newblock {\em Journal of Functional Analysis}, vol. 225, no. 1, pp. 229--245,
  2005.

\bibitem{LM2013}
M.~Ledoux and M.~Talagrand,
\newblock ``{Probability in Banach Spaces: isoperimetry and processes}''.
\newblock  Springer Science \& Business Media, 2013.

\bibitem{oymak2015sors}
S. Oymak, B. Recht, and M. Soltanolkotabi,
\newblock ``{Isometric sketching of any set via the Restricted Isometry
  Property},''
\newblock {\em arXiv preprint arXiv:1506.03521}, 2015.

\bibitem{XJ2018}
C. Xu and L. Jacques,
\newblock ``Quantized compressive sensing with rip matrices: The benefit of
  dithering,''
\newblock {\em arXiv preprint arXiv:1801.05870}, 2018.

\bibitem{CT06}
E.~J. Cand\`es and T. Tao,
\newblock ``{Near-optimal signal recovery from random projections: Universal
  encoding strategies?},''
\newblock {\em IEEE transactions on information theory}, vol. 52, no. 12, pp.
  5406--5425, 2006.

\bibitem{BB2008}
P.~T. Boufounos and R.~G. Baraniuk,
\newblock ``{1-bit compressive sensing},''
\newblock in {\em {Information Sciences and Systems, 2008. CISS 2008. 42nd
  Annual Conference on}}. IEEE, 2008, pp. 16--21.

\bibitem{PV2013}
Y. Plan and R. Vershynin,
\newblock ``{Robust 1-bit compressed sensing and sparse logistic regression: A
  convex programming approach},''
\newblock {\em IEEE Transactions on Information Theory}, vol. 59, no. 1, pp.
  482--494, 2013.

\bibitem{JDV2015}
L.~Jacques, K.~Degraux, and C.~De Vleeschouwer,
\newblock ``{Quantized Iterative Hard Thresholding: Bridging 1-bit and
  High-Resolution Quantized Compressive Sensing},''
\newblock in {\em {Proc. of SAMPTA2013 (July 1st-5th, Bremen, Germany)}}. IEEE,
  2013, pp. 105--108.

\bibitem{OptQuanLass2016}
X. Gu, S. Tu, H.-J.~Michael Shi, M. Case, D. Needell, and
  Y. Plan,
\newblock ``{Optimizing quantization for Lasso recovery},''
\newblock {\em arXiv preprint arXiv:1606.03055}, 2016.

\bibitem{PV2016}
Y. Plan and R. Vershynin,
\newblock ``{The generalized Lasso with non-linear observations},''
\newblock {\em IEEE Transactions on information theory}, vol. 62, no. 3, pp.
  1528--1537, 2016.

\bibitem{PVY2017}
Y. Plan, R. Vershynin, and E. Yudovina,
\newblock ``{High-dimensional estimation with geometric constraints},''
\newblock {\em Information and Inference: A Journal of the IMA}, vol. 6, no. 1,
  pp. 1--40, 2017.

\bibitem{JLBB2013}
L. Jacques, J.~N. Laska, P.~T. Boufounos, and R.~G. Baraniuk,
\newblock ``{Robust 1-bit compressive sensing via binary stable embeddings of
  sparse vectors},''
\newblock {\em IEEE Transactions on Information Theory}, vol. 59, no. 4, pp.
  2082--2102, 2013.

\end{thebibliography}
\end{document}